\documentclass[submission,copyright,creativecommons]{eptcs}
\providecommand{\event}{FROM 2023}
\providecommand{\ffrac}[2]{\ensuremath{\frac{\underline{#1}}{#2}}}

\usepackage{iftex}
\usepackage{amssymb,amsmath,amsthm,cmll}
\usepackage{listings,color,xcolor}
\hypersetup{hidelinks}

\newtheorem{theorem}{Theorem}

\ifpdf
  \usepackage{underscore}         
  \usepackage[T1]{fontenc}        
\else
  \usepackage{breakurl}           
\fi

\title{Using Z3 to Verify Inferences in Fragments of Linear Logic}
\author{Alen Docef \qquad\qquad Radu Negulescu
\institute{Department of Electrical and Computer Engineering\\
Virginia Commonwealth University\\
Richmond, Virginia, USA}
\email{adocef@vcu.edu \quad\qquad rnegulescu@vcu.edu}
\email{\hspace{100pt}rnegules@gmail.com}
\and
Mihai Prunescu
\institute{
Research Center for Logic, Optimization and Security\\
Faculty of Mathematics and Computer Science\\
University of Bucharest, Bucharest, Romania}
\institute{Simion Stoilow Institute of Mathematics\\
Romanian Academy, Bucharest, Romania}
\email{mihai.prunescu@imar.ro}
\email{mihai.prunescu@gmail.com}
}

\def\titlerunning{Using Z3 to Verify Inferences in Fragments of Linear Logic}
\def\authorrunning{A. Docef, R. Negulescu \& M. Prunescu}
\date{July 2023}

\begin{document}
\maketitle

\begin{abstract}
Linear logic is a substructural logic proposed as a refinement of classical and intuitionistic logics, with applications in programming languages, game semantics, and quantum physics. We present a template for Gentzen-style linear logic sequents that supports verification of logic inference rules using automatic theorem proving. Specifically, we use the Z3 Theorem Prover \cite{DeMoura2008} to check targeted inference rules based on a set of inference rules that are presumed to be valid. To demonstrate the approach, we apply it to validate several derived inference rules for two different fragments of linear logic: MLL+Mix (Multiplicative Linear Logic extended with a Mix rule) and MILL (Multiplicative Intuitionistic Linear Logic).

\begin{flushleft}
{\bf Keywords}: linear logic, MLL+Mix, MILL, Z3, inference rules\\
{\bf M.S.C. Classification}: 03B47, 03F52 
\end{flushleft}
\end{abstract}

\section{Introduction}

The Z3 Theorem Prover \cite{DeMoura2008} is a satisfiability modulo theories (SMT) solver targeted at software verification and program analysis. Besides SMT, the symbolic reasoning engine of Z3 also uses automatic reasoning, incremental solving, model generation, and other artificial intelligence techniques to determine satisfiability of a set of rules in a theory and to produce models. 

Linear logic \cite{Girard1987} is a substructural logic proposed as a refinement of classical and intuitionistic logics, with applications in programming languages, game semantics, and quantum physics. In \cite{mellies2022functorial}, the author makes a functorial connection between arbitrary models of multiplicative linear logic and the category of presheaves over arbitrary rings. Other far-reaching considerations connected with category theory are made by the same author in \cite{mellies2009categorical}. A more accessible approach to this connection
is presented in \cite{Schalk}. Connections with semantics for higher order quantum computing were studied in \cite{pagani2014applying}. A usual interpretation of linear logic, already intended by Girard, is that formulas do not hold values as true and false, but contain information about the availability and use of given resources. In this context, \cite{miller2004overview} presents an overview of linear logic programming. The article \cite{Baez2010} sketches a unified approach, a ``Rosetta stone'', based on categories as well, for interpreting linear logic in three seemingly unrelated domains: topology, quantum physics, and lambda calculus. Relations between linear logic and concurrency theory are an active area of research, for instance in \cite{aschieri2020par}. General presentations of linear logic can be found in: \cite{dicosmo2006linear}, \cite{dicosmoNDlinear}, \cite{girard1989proofs}, \cite{lafont1993introduction}, \cite{troelstra1996basic}, \cite{troelstra1992lectures}.

Two important fragments of linear logic are  multiplicative intuitionistic linear logic (MILL) and multiplicative linear logic with the Mix-rule (MLL+Mix). MILL is crystallized in \cite{benton1992term}, where its proof-theory is studied from a categorical theoretic point of view. A variant of MILL and its proof methods are discussed recently in \cite{gheorghiu2023proof}. MLL+Mix is crystallized in \cite{Abramsky1994} as a result of the authors' game-theoretic research in linear logic. They considered formulas as games and proofs as winning strategies. 

As observed for instance in \cite{chudigiewitsch2021computational}, in linear logic some usual proof steps, such as weakening and contraction, are restricted. The proof complexity in general linear logic is $\Sigma^1 _0$-hard. Some fragments are better behaved, the proof complexity in the multiplicative fragment being NP-complete. 
Still, proofs in linear logic are recognized to be computationally difficult.

We propose a template for modeling Gentzen-style sequents that supports verification of logic inference rules using Z3. Using our template, Z3 checks targeted inference rules based on a set of inference rules presumed to be valid. If a targeted rule does not hold, Z3 can produce a model assignment which serves as a counterexample.

We show how to apply our template to validate derived inference rules in two different fragments of linear logic. We find that this approach can flatten the learning curve when switching from classical logic to more general types of logic. Such transitions can be confusing. In particular, when concrete problems involving resource management are encoded in the satisfiability check of linear logic formulas, there is value added in using a theorem prover to offset the lack of common mathematical intuition. 

Theorem provers have been used to assist in generating Gentzen-style
proofs, for example Z3 in \cite{mcmillan2011} and Lean in \cite{patterson2023}.
However, we are not aware of a previous attempt to generalize
the approach using a template for modeling Gentzen sequents.
Our choice for Z3 was motivated by convenience and we hope it will
offer a benchmark for developing Lean tactics for MILL and MLL+Mix
as well.

\section{Modeling Inference Rules}

A {\it multiset} is a set with multiplicities: the same element can occur several times in a multiset, and the number of occurrences is called multiplicity. While the classical multiplicity of some element in a set is $1$, in a multiset elements may have as a multiplicity every set-cardinality. As our multisets of formulas will be later identified with formulas, we will consider only finite multisets. Implicitly, multiplicities are also finite. 

In the sequent notation, $\Gamma$, $\Delta$, etc., stand for multisets of formulas. $A$ and $B$ represent formulas. The turnstile symbol ($\vdash$, read `entails') separates the context (antecedents) from the conclusion (consequent). Generally, an inference rule in sequent calculus, including linear logic, takes the following form:
\[
\frac{\Gamma_1 \vdash A_1 \qquad \Gamma_2 \vdash A_2 \quad\ldots\quad \Gamma_n \vdash A_n}{\Gamma \vdash A}
\; \textit{(Name of Rule)}
\]
Above the line, we have the assumptions or premises of the rule, and below the line is the guarantee or conclusion of the rule, which holds true if all the antecedents are true. Each $\Gamma_i \vdash A_i$ is a sequent, where the context $\Gamma_i$ can be understood as a set of additional assumptions. Such a rule can be stated equivalently using Boolean logic:
\[
(\Gamma_1 \vdash A_1) \land (\Gamma_2 \vdash A_2) \land \ldots \land (\Gamma_n \vdash A_n) \rightarrow (\Gamma \vdash A),
\]
where $\land$ and $\rightarrow$ are logical conjunction and implication, respectively, meaning: "If $\Gamma_1 \vdash A_1$ and $\Gamma_2 \vdash A_2$ and $\ldots$ and $\Gamma_n \vdash A_n$ then $\Gamma \vdash A$".

For example, the tensor rule in linear logic often has the following form \cite[p.40]{Baez2010}:
\[
\frac{\Gamma \vdash A \qquad \Delta \vdash B}{\Gamma \otimes \Delta \vdash A \otimes B}
\; (\otimes)
\]
This rule states that if from context $\Gamma$ we can deduce $A$, and from context $\Delta$ we can deduce $B$, then from the context `$\Gamma$ tensor $\Delta$' ($\Gamma \otimes \Delta$) we can deduce `$A$ tensor $B$' ($A \otimes B$). The same tensor rule is sometimes stated in an alternate form \cite[p.3]{Heijltjes2016}:
\[
\frac{\Gamma,A \qquad \Delta,B}{\Gamma,\Delta, A \otimes B}
\; (\otimes) \qquad \textrm{meaning} \qquad
\frac{\vdash\Gamma,A \qquad \vdash\Delta,B}{\vdash \Gamma,\Delta, A \otimes B}
\; (\otimes)
\]

When the antecedent is missing, the turnstile symbol ($\vdash$), usually read as `entails', can be read as `provable'. That is, the consequent is provable from given rules, meaning the consequent is entailed by an empty context. In such cases, the turnstile can be omitted for legibility.

For the purpose of verification in Z3, we use its Python API to model inference rules as follows. First, we import the Z3 module and declare several objects used to model inference rules:
\begin{itemize} \setlength\itemsep{0em}
\item A solver, which will collect all the rules and verify their consistency.
\item A sort {\tt F} for formulas and contexts (i.e. multisets of formulas). Sorts are the Z3 model for types.
\item A function `entails' which takes two {\tt F}s (formulas or contexts) and returns a Boolean signifying that the left formula or context entails the right formula or context.
\item Two operations of linear logic (`tensor' and `lollipop') as functions of {\tt F}s returning {\tt F}s.
\item Three variables (called {\it constants} in Z3) of type {\tt F}.
\end{itemize}
The relevant code fragment is given below. The complete code for this section is provided in Listing~\ref{lst:generic}.

\begin{lstlisting}[numbers=none]
## Importing the z3 module
from z3 import *

## Declarations
ll = Solver()
F = DeclareSort('F')
entails = Function('entails', F, F, BoolSort())
tensor = Function('tensor', F, F, F)
lpop = Function('lollipop', F, F, F)
x, y, z = Consts('x y z', F)
\end{lstlisting}

The operators we declared satisfy inference rules. In our model, each inference rule is universally quantified by the variables in it. Then follows an expression that returns a Boolean. The expression is in prefix notation, i.e., operator followed by operands. Some expressions use {\tt Implies} and {\tt And}; these operators, predefined in Z3, have Boolean parameters and return values.


Consider a simple example where there are only two given inference rules:
\[ {X \vdash X} \qquad \textrm{and} \qquad (X \otimes Y \; \vdash \; Z) \leftrightarrow (Y \vdash X \multimap Z) \]
These are a subset of the inference rules in MILL, used here for illustration only; we will revisit these two rules later in Section 4. We add these two rules to our {\tt ll} solver by invoking its {\tt .add} method:

\begin{lstlisting}[numbers=none]
## Given rules
ll.add(ForAll([x], entails(x, x))) # (i)
ll.add(ForAll([x,y,z], entails(tensor(x,y),z) == entails(y,lpop(x,z)))) # (c)
\end{lstlisting}

To prove that a derived inference rule is valid, we add its {\it negation} to the existing set of rules and we check the satisfiability of the expanded set of rules. Unsatisfiability proves that the new rule is valid. In our template, checking for satisfiability is implemented using the {\tt ll.check()} method. Upon executing this check, Z3 prints either {\tt `sat'} or {\tt `unsat'} depending on whether a model is found that satisfies the set of rules, including the negated rule.

\begin{itemize}
\item If Z3 prints {\tt `unsat'}, the proposed rule results from the previously added rules. To see that, let $r_1, \ldots, r_n$ be the given rules, and let $r$ be the new rule prior to negation. If the set of rules is unsatisfiable that means we have 
\[
 \lnot (r_1 \land r_2 \land \ldots \land r_n \land \lnot r)
\]
Therefore, $\lnot (r_1 \land r_2 \land \ldots \land r_n) \lor r$. Thus, $(r_1 \land r_2 \land \ldots \land r_n) \rightarrow r$. In short, if Z3 prints {\tt `unsat'}, we have $(r_1 \land r_2 \land \ldots \land r_n) \rightarrow r$.
\item If Z3 prints {\tt `sat'}, it found a model, an assignment for the variables in the rules, under which the new rule $r$ is false (its negation is true) while the previously added rules are all true. Thus, if Z3 prints {\tt `sat'}, we do not have $(r_1 \land r_2 \land \ldots \land r_n) \rightarrow r$. If Z3 prints {\tt `sat'}, we can also use {\tt ll.print()} to obtain the model that caused the satisfiability, that is, under which assignment of variables did the previous rules and the negation of the new rule hold true.
\end{itemize}

For example, this is the modus ponens rule in linear logic: $x \otimes (x \multimap y) \vdash y$. To show that this rule follows from the two given rules, we add its negation to the solver. Running the satisfiability check as described above, we obtain an {\tt `unsat'} result, thus proving the derived rule. The relevant code fragment is:

\begin{lstlisting}[numbers=none]
## Derived rules
ll.add(Not(ForAll([x,y], entails(tensor(x,lpop(x,y)),y))))

## Verification
print(ll.check())
print(ll.model())
\end{lstlisting}

It is important to verify consistency of the given inference rules without the new negated rule. An inconsistent set would produce {\tt `unsat'} results even without the new negated rule. As a result, the given set of inference rules, if inconsistent, would imply the new negated rule because {\tt false} implies anything. Having commented out the negated modus ponens rule, we verify that the given set of rules is consistent by running a satisfiability check. The solver produces a {\tt `sat'} result, indicating that the given set of rules is indeed consistent.

\section{Verifying MLL+Mix Properties}

Multiplicative Linear Logic (MLL) \cite{Heijltjes2016} is a fragment of linear logic \cite{Abramsky1994} that restricts the linear logic structure to the multiplicative operators. MLL+Mix is an extension of MLL with the Mix inference rule \cite{Girard1987}. MLL+Mix has been applied to game theory as shown in \cite[p.546, Table 1]{Abramsky1994}.
 
For the MLL inference rules we follow the notation in \cite{Heijltjes2016}. For each group of formulas in \cite[Fig.1 p.3]{Heijltjes2016} there is an implicit entails symbol with an empty antecedent context. For example, the conclusion of formula {\it (ax)} is written $A, A^\star$ representing actually $\vdash A, A^\star$. To represent more easily entails with an empty antecedent, we define a unary predicate {\tt provable} on our sort {\tt F}, that is, a function from {\tt F} to Booleans, as follows (for the complete code see Listing~\ref{lst:mllmix}):

\begin{lstlisting}[numbers=none]
## Declarations
provable = Function('provable', F, BoolSort())
\end{lstlisting}

After restoring the entails symbols, the given inference rules of MLL+Mix are rephrased as follows.
\begin{alignat*}{4}
\frac{\vdash\Gamma}{\vdash\Gamma,\bot} &\;(\bot) \quad \quad \quad&
\frac{}{\vdash 1} &\;(1) \quad \quad \quad &
\frac{\vdash\Gamma,A,B}{\vdash\Gamma,A \parr B} &\;(\parr) \quad \quad \quad \\
\frac{\vdash\Gamma,A \qquad \vdash\Delta,B}{\vdash\Gamma,\Delta,A \otimes B} &\;(\otimes) &
\frac{}{\vdash A,A^\star} &\;(ax) \quad \quad \quad &
\frac{\vdash\Gamma,A \qquad \vdash\Delta,A^\star}{\vdash\Gamma,\Delta} &\;(cut)
\end{alignat*}

For the Mix rule, we use the formula from \cite[p.546, Table 1]{Abramsky1994}
\[
\frac{\vdash\Gamma \qquad \vdash\Delta}{\vdash\Gamma,\Delta} \;(mix)
\]

The given inference rules of MLL+Mix are modeled in Z3 as follows:
\begin{lstlisting}[numbers=none]
## Declarations
comma = Function('comma', F, F, F)
par = Function('par', F, F, F)
tensor = Function('tensor', F, F, F)
dual = Function('dual', F, F)
g, d, l, a, b, bot, one = Consts('g d l a b bot one', F)

## Axioms for comma (multiset reunion): associativity and commutativity
ll.add(ForAll([a, b, g], comma(a,comma(b,g)) == comma(comma(a,b),g)))
ll.add(ForAll([a, b], comma(a,b) == comma(b,a)))

## Given rules (Heijltjes and Houston page 3)
ll.add(ForAll([g],Implies(provable(g), provable(comma(g,bot)))))
ll.add(provable(one))
ll.add(ForAll([g, a, b],Implies(
    provable(comma(comma(g,a),b)),
    provable(comma(g,par(a, b))))))
ll.add(ForAll([g, a, d, b],Implies(
    And(provable(comma(g,a)), provable(comma(d,b))),
    provable(comma(g,comma(d,tensor(a, b)))))))
ll.add(ForAll([a],provable(comma(a,dual(a)))))
ll.add(ForAll([g, d, a],Implies(
    And(provable(comma(g,a)), provable(comma(d,dual(a)))),
    provable(comma(g,d)))))

## Mix rule (https://www.pls-lab.org/en/Mix_rule)
ll.add(ForAll([g, d],Implies(And(provable(g), provable(d)),
    provable(comma(g,d)))))
\end{lstlisting}

We also introduced an additional operator, called `comma', which adjoins two {\tt F}s into a new {\tt F}. It models the union of multisets of formulas, or multisets and single formulas. For this operator, we define two additional rules (associativity and commutativity). These rules ensure that contexts (i.e. multisets of formulas) created using the comma operator are equivalent regardless of the ordering of formulas in them.

Different authors use different notations for linear logic and its fragments. For example, the same rule (the $\otimes$ `tensor' rule) is presented differently in \cite[p.40]{Baez2010} and \cite[p.3]{Heijltjes2016}. However, the two formulations can be linked by interpreting `comma' as `par'. This link also provides the rationale for our merging the contexts and formulas in the same sort. As in \cite{Baez2010} a context (multiset of formulas) is seen here as the par of the formulas in that context. Informally, par ($\parr$) can be thought of as a parallel composition of devices \cite{aschieri2020par}. In what follows, we show that this viewpoint over the meaning of contexts supports the verification of properties in two different fragments of linear logic, following the notations in \cite{Baez2010} and \cite{Heijltjes2016}.

First, we verify that our model for the rules of MLL+Mix is indeed satisfiable. Indeed, {\tt ll.check()} returns {\tt `sat'} on our model of the rules of MLL+Mix, indicating that our rules are consistent.

As shown in Section 2, we can verify a new inference rule in MLL+Mix by adding its negation to the solver and verifying the satisfiability of the resulting set of rules. Using this approach, we verified, one by one, the validity of the following derived rules from \cite[Fig. 2 and Fig. 3]{Heijltjes2016}:
\[
\frac{}{\vdash 1,\bot} \; (mix1) \qquad\qquad
\frac{}{\vdash A \otimes B, A^\star \parr B^\star} \; (mix2) \qquad\qquad
\frac{\vdash\Gamma,A \qquad \vdash\Delta,B \qquad \vdash\Lambda,A^\star,B^\star}{\vdash \Gamma,\Delta,\Lambda} \; (mix3)
\]
The relevant code is as follows:
\begin{lstlisting}[numbers=none]
## Derived rules (Heijltjes and Houston page 4, all "unsat")
# ll.add(Not(provable(comma(one,bot))))
# ll.add(Not(ForAll([a,b],
#     provable(comma(tensor(a,b), par(dual(a),dual(b)))))))
# ll.add(Not(ForAll([a,b,g,d,l],
#     Implies(
#         And(And(provable(comma(g,a)),provable(comma(d,b))),
#             provable(comma(comma(l,dual(a)),dual(b)))),
#         provable(comma(comma(g,d),l))))))
\end{lstlisting}

We also verify several alternate formulations to the Mix rule in the context of MLL, thought to be equivalent to Mix. We find that these formulations are not equivalent to Mix. One at a time, we check the following negated rules are satisfiable in the context of MLL + Mix, indicating that the respective rules do not hold.
\[
\vdash 1 \leftrightarrow \bot \qquad
\vdash ( \bot^\star \parr 1 ) \qquad
\vdash \bot
\]
The relevant code is:
\begin{lstlisting}[numbers=none]
## Invalid derived rules (all "sat")
# ll.add(Not(one == bot))
# ll.add(Not(provable(par(dual(bot),one))))
# ll.add(Not(provable(bot)))
\end{lstlisting}

To sum up the results of our investigations into MLL+Mix inference rules, we can state the following Theorem:

\begin{theorem}
\,\,\,\,{\it The following sequents (and rule) are provable in} MLL+Mix:
\begin{alignat*}{1}
\vdash 1,\bot \qquad
\vdash A \otimes B, A^\star \parr B^\star \qquad
\frac{\vdash\Gamma,A \qquad \vdash\Delta,B \qquad \vdash\Lambda,A^\star,B^\star}{\vdash \Gamma,\Delta,\Lambda}
\end{alignat*}
{\it The following sequents are not provable in} MLL+Mix:
\begin{alignat*}{1}
\vdash 1 \leftrightarrow \bot \qquad
\vdash \bot^\star \parr 1 \qquad
\vdash \bot
\end{alignat*}
\end{theorem}

The theory MLL+Mix can be substantially improved by adding some other axioms. Consider, as an example, the following form of Contraction $(C)$:
\[\frac{\vdash 1,A}{\vdash A} \;(C)\]

We add the Contraction rule $(C)$ to the solver as follows:
\begin{lstlisting}[numbers=none]
## Contraction (subset) rule and resulting theorem
# ll.add(ForAll([a, b], Implies(provable(comma(one,a)),provable(a))))
# ll.add(Not(provable(bot))) # unsat
\end{lstlisting}

The new theory MLL+Mix+C proves to be consistent by using Z3. This extension leads to a strictly stronger theory, as shown by Theorem 2 below:

\begin{theorem}
The theory MLL+Mix+C proves the sequent: \[\vdash \bot.\]
\end{theorem}
\begin{proof}
From the rule $(1)$ we get \[\vdash 1.\]
From the rule $(\bot)$, we get \[\vdash 1, \bot.\]
Finally, from the rule $(C)$ we get \[\vdash \bot.\] 
\end{proof}

It is however not certain that this extended theory is still a fragment of Linear Logic. Another possibility is to consider the following axiom, called $(\emptyset)$, which introduces sequents with empty hypothesis and empty conclusion:
\[\frac{}{\,\,\,\,\,\vdash\,\,\,\,\,} \,\,(\emptyset)\]
It is clear that MLL+Mix+$(\emptyset)$ proves $\vdash \bot$ if one also introduces the convention \[\forall\,x\,\,\,\,{\rm Comma}(\emptyset,x) = x.\]
By $(\emptyset)$ we get the empty sequent $\,\,\,\,\vdash\,\,\,\,$. By $(\bot)$, we get $\vdash ,\bot\,\,$ and applying the convention, this is $\vdash \bot$. The Z3 implementation of the new rule $(\emptyset)$ remains to be completed.

\section{Verifying MILL Properties}

A fragment of Linear Logic called Multiplicative Intuitionistic Linear Logic (MILL) is shown in \cite{Baez2010} to be closely related to models of topology, quantum physics, and lambda calculus. Using the technique in Section 2, we model MILL in Z3 and we investigate several properties of MILL. The complete code is provided in Listing~\ref{lst:mill}.

Our starting point is the set of rules {\it (i), (o), ($\otimes$), (a), (l), (r), (b), (c)} from \cite[p.40]{Baez2010}.
\begin{alignat*}{2}
\frac{}{X \vdash X} &\;(i) \qquad\qquad\qquad&
\frac{X \vdash Y \quad Y \vdash Z}{X \vdash Z} &\;(o)\\
\frac{W \vdash X \quad Y \vdash Z}{W \otimes Y \vdash X \otimes Z} &\;(\otimes)&
\ffrac{W \vdash (X \otimes Y) \otimes Z}{W \vdash X \otimes (Y \otimes Z)} &\;(a)\\
\ffrac{X \vdash I \otimes Y}{X \vdash Y} &\;(l)&
\ffrac{X \vdash Y \otimes I}{X \vdash Y} &\;(r)\\
\ffrac{W \vdash X \otimes Y}{W \vdash Y \otimes X} &\;(b)&
\ffrac{X \otimes Y \; \vdash \; Z}{Y \vdash X \multimap Z} &\;(c)
\end{alignat*}

Rule {\it (i)} with no premises is modeled in Z3 by simply adding {\tt entails(x,x)}, universally quantified. There is no need for an {\tt Implies} in this rule because the empty set of premises means the conjunction of premises in that set is true. Thus, rule {\it (i)} is $(\mathrm{true} \rightarrow (x \vdash x)$), which is equivalent to $(\lnot \mathrm{true} \lor (x \vdash x))$, which is equivalent to just $(x \vdash x)$. Finally, we model rule {\it (i)} by adding to the solver an expression with no Boolean logic operators. Rule {\it (i)}, as well as rules {\it (o)} and ($\otimes$) are modeled using the approach described in Section 2. The relevant code is:

\begin{lstlisting}[numbers=none]
## Given rules
ll.add(ForAll([x], entails(x, x))) # rule (i)
ll.add(ForAll([x,y,z], Implies(And(
    entails(x,y), entails(y,z)), entails(x, z)))) # rule (o)
ll.add(ForAll([w,x,y,z], Implies(And(
    entails(w,x), entails(y,z)), entails(tensor(w, y), tensor(x, z)))))
    # rule (tensor)
\end{lstlisting}

The double bar represents implication both ways, i.e., logical equivalence. Accordingly, rules {\it (a), (l), (r), (b), (c)} are modeled essentially by the technique in Section 2, except that instead of {\tt Implies} we use {\tt ==} for logic equivalence. The relevant code for these rules is:

\begin{lstlisting}[numbers=none]
## Given rules
ll.add(ForAll([w,x,y,z],
    entails(w,tensor(tensor(x,y),z)) == entails(w, tensor(x, tensor(y, z)))))
    # rule (a)
ll.add(ForAll([x,y],
    entails(x,tensor(I,y)) == entails(x,y))) # rule (l)
ll.add(ForAll([x,y],
    entails(x,tensor(y,I)) == entails(x,y))) # rule (r)
ll.add(ForAll([w,x,y],
    entails(w,tensor(x,y)) == entails(w,tensor(y,x)))) # rule (b)
ll.add(ForAll([x,y,z],
    entails(tensor(x,y),z) == entails(y,lpop(x,z)))) # rule (c)
\end{lstlisting}

After adding to the solver the eight rules of MILL, we attempt to verify several new inference rules. We start with two properties from \cite[p.41]{Baez2010}: modus ponens and internal composition. Ther relevant code for both rules is:

\begin{lstlisting}[numbers=none]
## Derived rules (all "unsat")
# ll.add(Not(ForAll([x,y], entails(tensor(x,lpop(x,y)),y)))) # rule (ev)
# ll.add(Not(ForAll([x,y,z], entails(tensor(lpop(x,y),lpop(y,z)),
#                                    lpop(x,z))))) # internal composition rule
\end{lstlisting}

The {\em modus ponens} rule is
\[ \frac{}{X \otimes (X \multimap Y) \vdash Y} \;(ev) \]
One way to interpret rule {\it (ev)} is by taking tensor ($\otimes$) to mean linear conjunction and lollipop ($\multimap$) to mean linear implication. Therefore, if we have X linear-and X linear-implies Y, we have Y. To verify rule {\it (ev)} we model it using the same technique as for rule (i) above, i.e. by adding its negation to the set of given rules and checking the satisfiability of the extended set of rules. Since the result is {\tt `unsat'}, we conclude that rule {\it (ev)} holds in MILL.

Next we verify a more complex rule, called {\em internal composition} in \cite[p.41]{Baez2010}:
\[ (X\multimap Y) \otimes (Y\multimap Z) \vdash (X\multimap Z) \]
Informally, taking $\otimes$ for linear-and and $\multimap$ for linear-implication, the internal composition rule says that we can compose chained linear-implications. To verify internal composition in MILL, we add the internal composition rule, negated. We also comment out the modus ponens rule so it won't interfere with the satisfiability check. Running the satisfiability check again, we obtain {\tt `unsat'}, meaning that the internal composition rule holds in MILL.

Two important operators in linear logic are dual and par, denoted here by $\star$ and $\parr$, respectively, following \cite{Heijltjes2016}. Rules {\it (i)} through {\it (c)} above do not include rules for dual and par, but we can construct dual and par by new rules as follows:
\begin{alignat*}{2}
&\textrm{Dual:} \qquad & X^\star   & = X \multimap I \\
&\textrm{Par:}         & X \parr Y & = X^\star \multimap Y
\end{alignat*}
Here $I$ is a constant which is the unit of tensor; its properties are specified in rules {\it (l)} and {\it (r)}. These rules are implemented as follows:
\begin{lstlisting}[numbers=none]
## Given rules for dual and par
ll.add(ForAll([x], dual(x) == lpop(x, I))) # dual
ll.add(ForAll([x,y], par(x, y) == lpop(dual(x), y))) # par
\end{lstlisting}

Now we can verify a few properties of dual and par in MILL by adding the respective negated rules to the solver, as shown in the code fragment below. We emphasize that these properties are verified one at a time by un-commenting its corresponding line while commenting all other properties.
\begin{lstlisting}[numbers=none]
## Proving properties of dual and par (all "unsat")
# ll.add(Not(ForAll([x], entails(x, dual(dual(x)))))) # x |- x dual dual
# ll.add(Not(ForAll([x,y], entails(par(dual(x), y), lpop(x, y)))))
#     # X dual par Y |- X -o Y
# ll.add(Not(ForAll([x,y],
#     entails(par(x, y), dual(tensor(dual(x), dual(y)))))))
#     # x par y |- - (- x tensor - y)
#
## out of memory
# ll.add(Not(ForAll([x], entails(dual(dual(x)), x)))) # x dual dual |- x
# ll.add(Not(ForAll([x,y], entails(lpop(x, y), par(dual(x), y)))))
#     # X -o Y |- X dual par Y
# ll.add(Not(ForAll([x,y],
#     entails(dual(tensor(dual(x), dual(y))), par(x, y)))))
#     # - (- x tensor - y) |- x par y
\end{lstlisting}

First we want to know to what extent dual is an involution in MILL, so we verify the rules $X^{\star\star} \vdash X$ and $X \vdash X^{\star\star}$ by running the solver check method. The solver yields {\tt `unsat'} for the second rule, proving that the rule holds in MILL. The solver runs out of memory for the first rule. Details about the hardware and software used to execute the prover will be given later in the paper.
 
The next two rules to be checked express relationships between par ($\parr$) and linear-implies ($\multimap$):
\[ X^\star \parr Y \vdash X\multimap Y \qquad \textrm{and} \qquad X\multimap Y \vdash X^\star \parr Y \]
Adding, one at a time, the two negated rules, the solver yielded {\tt `unsat'} for the first rule and ran out of memory for the second rule. This means that the first rule holds in MILL and a counterexample could not be easily found in MILL for the second rule.

Now we verify relationships between par ($\parr$) and tensor ($\otimes$):
\[ X \parr Y \vdash (X^\star \otimes Y^\star)^\star \qquad \textrm{and} \qquad (X^\star \otimes Y^\star)^\star \vdash X \parr Y \]
The first rule holds (the solver yields {\tt `unsat'}), and for the second rule the solver runs out of memory, as above.

Finally, we prove that $I$ is self-dual. We also check to what extent $I$ is a neutral element for the par ($\parr$) operator.
\[I^\star \vdash I \qquad \textrm{and} \qquad I \vdash I^\star  \]
\[ X \parr I \vdash X \qquad \textrm{and} \qquad X \vdash X \parr I \]
The respective negated rules are implemented as shown in the code below. Three of the four rules hold (the solver yields {\tt `unsat'}).  For the rule $X \parr I \vdash X$, the solver runs out of memory, and we cannot prove anything.

\begin{lstlisting}[numbers=none]
## Proving properties of dual and par (all "unsat")
# ll.add(Not(entails(dual(I),I)))
# ll.add(Not(entails(I,dual(I))))
# ll.add(Not(ForAll([x],entails(x,par(x,I)))))
#
## out of memory
# ll.add(Not(ForAll([x],entails(par(x,I),x))))
\end{lstlisting}

Recalling the intuitionistic nature of this logic, some of the propositions that we could not prove or disprove by Z3 are unlikely to be consequences of the theory. It is however remarkable that it is so difficult for Z3 to construct counterexamples. We believe that some of the propositions that caused Z3 to run out of memory are actually consistent with the theory MILL.

To sum up the results of our investigations into MILL inference rules, we restate the derived inference rules by means of a Theorem. We use the ($\dashv\vdash$) symbol to denote bidirectional entailment.

\begin{theorem}
\,\,\,\,The following sequents are provable in MILL:
\[
\begin{array}{r@{\;}c@{\;}l}
X \otimes (X \multimap Y) & \vdash & Y \\
(X\multimap Y) \otimes (Y\multimap Z) & \vdash & (X\multimap Z) \\
X & \vdash & X^{\star\star} \\
X^\star \parr Y  & \vdash & X\multimap Y \\
X \parr Y & \vdash & (X^\star \otimes Y^\star)^\star \\
X  & \vdash & X \parr I \\
I^\star & \dashv\vdash & I \\
\end{array}
\]
\end{theorem}

Finally, a few details of the software and hardware platforms we used. 
The experiments have been run on
a high performance system with four 10-core Intel Xeon E7-4870 with 2.4 GHz clock, 30 MB cache, and 1 TB total RAM. The software is Python version 3.9.13 running in a Rocky Linux operating system, version 8.7. Successful experiments completed in a few seconds. Experiments that ran out of memory never completed.


\section{Conclusions}

Beyond the inference rules in MILL and MLL+Mix, our modeling exercise has led us to develop rules for handling comma (union of multisets) and entails with empty antecedent context.

The exercise has shown that several known properties of MILL and MLL+Mix can be proven, as expected. It also clarifies that MLL rules are insufficient for several other inference rules to be derived from Mix. Several rules that we checked caused Z3 to fill up the available memory and run out of time as a result.

Following \cite{LeanAndZ3}, further work may include using Z3 as a tactic in Lean, another theorem prover.
Lean uses different technologies than Z3 and the combination of Lean and Z3 can prove the rules that caused Z3 to run out of memory.

Successful validation of the derived inference rules suggests that the proposed practical approach has the potential to be more widely applicable in other formal logic systems, as Linear Temporal Logic (see e.g. \cite{grigorerosu}, \cite{dimacatalin}), Reachability Logic (e.g. \cite{rosustefanescu}), IMLL (e.g. \cite{gheorghiu2023proof}) and so on.

\section{Appendix}

\begin{lstlisting}[caption={Python code for modeling inference rules},label={lst:generic}]
## Importing the z3 module
from z3 import *

## Declarations
ll = Solver()
F = DeclareSort('F')
entails = Function('entails', F, F, BoolSort())
tensor = Function('tensor', F, F, F)
lpop = Function('lollipop', F, F, F)
x, y, z = Consts('x y z', F)

## Given rules
ll.add(ForAll([x], entails(x, x))) # (i)
ll.add(ForAll([x,y,z], entails(tensor(x,y),z) == entails(y,lpop(x,z)))) # (c)

## Derived rules
ll.add(Not(ForAll([x,y], entails(tensor(x,lpop(x,y)),y))))

## Verification
print(ll.check())
print(ll.model())
\end{lstlisting}

\begin{lstlisting}[caption={Python code for modeling MLL+Mix inference rules},label={lst:mllmix}]
## Importing the z3 module
from z3 import *

## Declarations
ll = Solver()
F = DeclareSort('F')
entails = Function('entails', F, F, BoolSort())
provable = Function('provable', F, BoolSort())
comma = Function('comma', F, F, F)
par = Function('par', F, F, F)
tensor = Function('tensor', F, F, F)
dual = Function('dual', F, F)
g, d, l, a, b, bot, one = Consts('g d l a b bot one', F)

## Axioms for comma (multiset reunion): associativity and commutativity
ll.add(ForAll([a, b, g], comma(a,comma(b,g)) == comma(comma(a,b),g)))
ll.add(ForAll([a, b], comma(a,b) == comma(b,a)))

## Given rules (Heijltjes and Houston page 3)
ll.add(ForAll([g],Implies(provable(g), provable(comma(g,bot)))))
ll.add(provable(one))
ll.add(ForAll([g, a, b],Implies(
    provable(comma(comma(g,a),b)),
    provable(comma(g,par(a, b))))))
ll.add(ForAll([g, a, d, b],Implies(
    And(provable(comma(g,a)), provable(comma(d,b))),
    provable(comma(g,comma(d,tensor(a, b)))))))
ll.add(ForAll([a],provable(comma(a,dual(a)))))
ll.add(ForAll([g, d, a],Implies(
    And(provable(comma(g,a)), provable(comma(d,dual(a)))),
    provable(comma(g,d)))))

## Mix rule (https://www.pls-lab.org/en/Mix_rule)
ll.add(ForAll([g, d],Implies(And(provable(g), provable(d)),
    provable(comma(g,d)))))

## Derived rules (Heijltjes and Houston page 4, all "unsat")
# ll.add(Not(provable(comma(one,bot))))
# ll.add(Not(ForAll([a,b],
#     provable(comma(tensor(a,b), par(dual(a),dual(b)))))))
# ll.add(Not(ForAll([a,b,g,d,l],
#     Implies(
#         And(And(provable(comma(g,a)),provable(comma(d,b))),
#             provable(comma(comma(l,dual(a)),dual(b)))),
#         provable(comma(comma(g,d),l))))))

## Invalid derived rules (all "sat")
# ll.add(Not(one == bot))
# ll.add(Not(provable(par(dual(bot),one))))
# ll.add(Not(provable(bot)))

## Contraction (subset) rule and resulting theorem
# ll.add(ForAll([a, b], Implies(provable(comma(one,a)),provable(a))))
# ll.add(Not(provable(bot))) # unsat

## Verification
print(ll.check())
print(ll.model())
\end{lstlisting}

\newpage
\begin{lstlisting}[caption={Python code for modeling MILL inference rules},label={lst:mill}]
## Importing the z3 module
from z3 import *

## Declarations
ll = Solver()
F = DeclareSort('F')
entails = Function('entails', F, F, BoolSort())
par = Function('par', F, F, F)
tensor = Function('tensor', F, F, F)
lpop = Function('lollipop', F, F, F)
dual = Function('dual', F, F)
x, y, z, w, I = Consts('x y z w I', F)

## Given rules
ll.add(ForAll([x], entails(x, x))) # rule (i)
ll.add(ForAll([x,y,z], Implies(And(
    entails(x,y), entails(y,z)), entails(x, z)))) # rule (o)
ll.add(ForAll([w,x,y,z], Implies(And(
    entails(w,x), entails(y,z)), entails(tensor(w, y), tensor(x, z)))))
    # rule (tensor)
ll.add(ForAll([w,x,y,z],
    entails(w,tensor(tensor(x,y),z)) == entails(w, tensor(x, tensor(y, z)))))
    # rule (a)
ll.add(ForAll([x,y],
    entails(x,tensor(I,y)) == entails(x,y))) # rule (l)
ll.add(ForAll([x,y],
    entails(x,tensor(y,I)) == entails(x,y))) # rule (r)
ll.add(ForAll([w,x,y],
    entails(w,tensor(x,y)) == entails(w,tensor(y,x)))) # rule (b)
ll.add(ForAll([x,y,z],
    entails(tensor(x,y),z) == entails(y,lpop(x,z)))) # rule (c)

## Derived rules (all "unsat")
# ll.add(Not(ForAll([x,y], entails(tensor(x,lpop(x,y)),y)))) # rule (ev)
# ll.add(Not(ForAll([x,y,z], entails(tensor(lpop(x,y),lpop(y,z)),
#                                    lpop(x,z))))) # internal composition rule

## Given rules for dual and par
ll.add(ForAll([x], dual(x) == lpop(x, I))) # dual
ll.add(ForAll([x,y], par(x, y) == lpop(dual(x), y))) # par

## Proving properties of dual and par (all "unsat")
# ll.add(Not(ForAll([x], entails(x, dual(dual(x)))))) # x |- x dual dual
# ll.add(Not(ForAll([x,y], entails(par(dual(x), y), lpop(x, y)))))
#     # X dual par Y |- X -o Y
# ll.add(Not(ForAll([x,y],
#     entails(par(x, y), dual(tensor(dual(x), dual(y)))))))
#     # x par y |- - (- x tensor - y)
# ll.add(Not(entails(dual(I),I)))
# ll.add(Not(entails(I,dual(I))))
# ll.add(Not(ForAll([x],entails(x,par(x,I)))))

## out of memory
# ll.add(Not(ForAll([x], entails(dual(dual(x)), x)))) # x dual dual |- x
# ll.add(Not(ForAll([x,y], entails(lpop(x, y), par(dual(x), y)))))
#     # X -o Y |- X dual par Y
# ll.add(Not(ForAll([x,y],
#     entails(dual(tensor(dual(x), dual(y))), par(x, y)))))
#    # - (- x tensor - y) |- x par y
# ll.add(Not(ForAll([x],entails(par(x,I),x))))
# ll.add(Not(ForAll([x,y,z],
#     entails(par(tensor(x,y),z), tensor(par(x,z),par(y,z))))))
#     # distributivity par tensor
# ll.add(Not(ForAll([x,y,z],
#     entails(tensor(par(x,z),par(y,z)), par(tensor(x,y),z)))))
#     # distributivity par tensor

## Verification
print(ll.check())
print(ll.model())
\end{lstlisting}

\nocite{*}
\bibliographystyle{eptcs}
\bibliography{generic}
\end{document}